\documentclass[11pt]{article}
\usepackage{amsmath,amssymb,amsfonts,amsthm,epsfig}
\usepackage[usenames,dvipsnames]{xcolor}

\usepackage{bm,xspace}

\usepackage{times}

\usepackage{cancel}
\usepackage{fullpage}
\usepackage{liyang}
\usepackage{framed}
\usepackage{verbatim}
\usepackage{enumitem}
\usepackage{array}
\usepackage{multirow}
\usepackage{afterpage}
\usepackage{mathrsfs}
\usepackage{comment}  
\usepackage{caption}
\usepackage[normalem]{ulem}

\usepackage{caption} 

\usepackage{todonotes}

\makeatletter
\newtheorem*{rep@theorem}{\rep@title}
\newcommand{\newreptheorem}[2]{
\newenvironment{rep#1}[1]{
 \def\rep@title{#2 \ref{##1}}
 \begin{rep@theorem}\itshape}
 {\end{rep@theorem}}}
\makeatother
\theoremstyle{plain}



\newcommand{\ignore}[1]{}

\def\colorful{1}

\ifnum\colorful=1

\fi
\ifnum\colorful=0

\fi

\usepackage{boxedminipage}

\newreptheorem{theorem}{Theorem}
\newtheorem*{theorem*}{Theorem}
\newreptheorem{lemma}{Lemma}
\newreptheorem{proposition}{Proposition}
\newtheorem*{noclaim*}{Claim}

\newcommand{\ssum}{\mathrm{sum}}

\usepackage{dsfont}

\renewcommand{\N}{\mathds{N}}

\def\XX{\mathbf{X}}
\def\rr{\mathbf{r}}

\usepackage{cleveref}

\begin{document}

\title{Efficient average-case population recovery in the presence of insertions and deletions}

{
\author{
Frank Ban\\
UC Berkeley\\
fban@berkeley.edu
\and
Xi Chen\\
Columbia University\\
xichen@cs.columbia.edu
\and
Rocco A. Servedio \\
Columbia University\\
rocco@cs.columbia.edu
\and
Sandip Sinha \\
Columbia University\\
sandip@cs.columbia.edu
}
}

\maketitle

\thispagestyle{empty}


\begin{abstract}

A number of recent works have considered the \emph{trace reconstruction problem}, in which an unknown source string $x \in \zo^n$ is transmitted through a probabilistic channel which may randomly delete coordinates or insert random bits, resulting in a \emph{trace} of $x$.  The goal is to reconstruct the original string~$x$ from independent traces of $x$.  While the asymptotically best algorithms known for worst-case strings use $\exp(O(n^{1/3}))$ traces \cite{DOS17,NazarovPeres17}, several highly efficient algorithms are known \cite{PZ17,HPP18} for the \emph{average-case} version of the problem, in which the source string $x$ is chosen uniformly at random from $\zo^n$. In this paper we consider a generalization of the above-described average-case trace reconstruction problem, which we call \emph{average-case population recovery in the presence of insertions and deletions}.  In this problem, rather than a single unknown source string there is an unknown distribution over $s$ unknown source strings $x^1,\dots,x^s \in \zo^n$, and each sample given to the algorithm is independently generated by drawing some $x^i$ from this distribution and outputting an independent trace of $x^i$.
	
Building on the results of \cite{PZ17} and \cite{HPP18}, we give an efficient algorithm for the average-case population recovery problem in the presence of insertions and deletions. 
For any support size $1 \leq s \leq \smash{\exp ( \Theta(n^{1/3}) ) }$, for a $1-o(1)$ fraction of all $s$-element support sets $\{x^1,\dots,x^s\} \subset \zo^n$, for every distribution ${\cal D}$ supported on $\{x^1,\dots,x^s\},$ our algorithm can efficiently recover ${\cal D}$ up to total variation distance at most $\eps$ with high probability, given access to independent traces of independent draws from ${\cal D}$ as described above.
The running time of our algorithm is $\poly(n,s,1/\eps)$ and its sample complexity is $\poly ( s,1/\eps,\exp(\log^{1/3} n) ).$ This polynomial dependence on the support size $s$ is in sharp contrast with the \emph{worst-case} version of the problem (when $x^1,\dots,x^s$ may be any strings in $\zo^n$), in which the sample complexity of the most efficient known algorithm \cite{BCFSS19} is doubly exponential in $s$.
\end{abstract}


\newpage

\setcounter{page}{1}


\section{Introduction} \label{sec:intro}

{\bf Background:  Worst-case and average-case trace reconstruction.}
In the problem of \emph{trace reconstruction in the presence of insertions and deletions}, there is an unknown and arbitrary $n$-bit source string $x \in \zo^n$ and the goal is to reconstruct $x$ given access to independent \emph{traces} of $x$.  A \emph{trace} of $x$ is a copy that has been passed through a noise channel
which independently removes each bit of $x$ with some probability $q$ (the \emph{deletion rate}) and also independently inserts random bits according to some \emph{insertion rate} $q'$.\footnote{A detailed description of the channel is given in~\Cref{sec:preliminaries}.  Augmented variants of this insertion\hspace{0.04cm}/\hspace{0.04cm}deletion noise model can also be considered, for example allowing for bit-flips as well as insertions and deletions, but unlike deletions and insertions bit-flips can typically be handled in a straightforward fashion. In this paper we confine our attention to the insertion\hspace{0.04cm}/\hspace{0.04cm}deletion channel.}  Intuitively, the insertion-deletion channel (or even just the deletion channel with no insertions) is challenging to deal with because it is difficult to determine which coordinate of the source string (if any, if insertions are possible) is responsible for a given coordinate of a received trace.

The insertion\hspace{0.04cm}/\hspace{0.04cm}deletion trace reconstruction problem is motivated by connections to  recovery problems arising in biology (see e.g.~\cite{ADHR10,DR10,ABH14}) and has been the subject of considerable research, especially in recent years.  The worst-case version of this problem, in which the source string $x$ can be an arbitrary element of $\zo^n$, appears to be quite difficult even for small constant noise rates. In early work  \cite{BKKM04} gave an efficient algorithm that succeeds in the deletion-only model if the deletion rate $q$ is quite low, at most $O(1/n^{1/2+\eps}).$  Also in the \mbox{deletion-only} model, \cite{HMPW08} showed that $\exp(\tilde{O}(\sqrt{n}))$ many traces suffice for any constant deletion~rate $q$ bounded away from $1$.  More recently, this result was  improved in simultaneous and independent works of \cite{DOS17} and  \cite{NazarovPeres17}, each of which showed that for any constant insertion and deletion rates $q,q'$, $\smash{\exp(O(n^{1/3}))}$ traces suffice to reconstruct any $x \in \zo^n$. These algorithms, which~run in $\smash{\exp({O(n^{1/3})})}$ time, give the best results to date for the worst-case problem.  (On the lower bound side, recent work of \cite{HoldenLyons18} obtained an
$\tilde{\Omega}(n^{5/4})$ lower bound on the number of traces required from the deletion channel, improving an earlier $\Omega(n)$ lower bound due to  \cite{MPV14}. Later work of \cite{Cha19} improved this lower bound to $\tilde{\Omega}(n^{3/2})$.)

Since the worst-case trace reconstruction problem seems to be quite difficult, and since the assumption that the source string $x$ is completely arbitrary may be overly pessimistic in  various contexts, it is~natural to consider an \emph{average-case} version of the problem in which the source string $x$ is assumed to be drawn uniformly at random from $\zo^n$.  This average-case problem has been intensively studied, and interestingly it turns out to be significantly easier than the worst-case problem.  \cite{BKKM04} showed~that~for most source strings $x$, in the deletion-only setting only $O(\log n)$ traces suffice for deletion rates $q$ as~large as $O(1/\log n)$.  \cite{KM05} considered the insertion\hspace{0.04cm}/\hspace{0.04cm}deletion noise channel and obtained an $O(\log n)$-trace average-case algorithm for noise rates up to $O(1/\log^2 n)$, which was later improved to $O(1/\log n)$ in \cite{VS08}. \cite{HMPW08} were the first to give an efficient (using $\poly(n)$ traces) average-case algorithm, for the deletion-only model, that succeeds for some constant deletion rate (their algorithm could handle deletion rates up to about $q=0.07$).  Building on the worst-case results of \cite{DOS17} and \cite{NazarovPeres17}, a number of significantly stronger average-case results have recently been established.   \cite{PZ17} gave an average-case algorithm for the deletion-only problem which uses $\exp(O(\log^{1/2} n))$ many traces for any deletion rate $q<1/2.$  Improving on this, \cite{HPP18} gave an average-case algorithm which can handle both insertions and deletions at any constant rate and uses only $\exp(O(\log^{1/3} n))$ many traces.  (A simple reduction shows that any improvement on this sample complexity for the average-case problem would imply an improvement of the $\exp({O(n^{1/3}))}$-trace worst-case algorithms of \cite{DOS17} and~\cite{NazarovPeres17}.)

\medskip
\noindent
{\bf Beyond trace reconstruction: Population recovery from the deletion channel.}
Inspired by a related problem known as \emph{population recovery}, recent work of \cite{BCFSS19} has considered a challenging extension of the trace reconstruction problem.  Population recovery is the problem of learning an unknown \emph{distribution} over an unknown set of $n$-bit strings, given access to independent draws from the distribution that have been independently corrupted according to some noise channel.  Most research in population recovery has focused on two noise models, namely the \emph{bit-flip} noise channel (in which each coordinate is independently flipped with some fixed probability) and the \emph{erasure} noise channel (in which each coordinate is independently replaced by `?' with some fixed probability), both of which have been intensively studied, see e.g.~\cite{DRWY12,WY16,PSW17,DOS17poprec,MoitraSaks13,LZ15,DST16,DOS17poprec}.    \cite{BCFSS19} considered the problem of \emph{population recovery from the deletion channel}.  This is a generalization of the deletion-channel trace reconstruction problem:  now there is an unknown distribution over $s$ unknown source strings $x^1,\dots,x^s \in \zo^n$, and each sample provided to the learner is obtained by first drawing
  a string $x^i$ from this distribution and then passing it through the  deletion  noise channel.
It is clear that this problem is at least as difficult as the trace reconstruction problem (which is the $s=1$ case), and indeed having multiple source strings turns out to pose significant new challenges. \cite{BCFSS19} considered the worst-case version of this problem, and showed that any distribution ${\cal D}$ over any set of $s$ unknown source strings can be recovered to total variation distance $\eps$ given $2^{\sqrt{n} \cdot (\log n)^{O(s)}}/\eps^2$
many traces from the deletion channel. \cite{BCFSS19} also gave a lower bound, showing that for any $s \leq n^{0.49}$ at least $n^{\Omega(s)}$ many traces are required. {Population recovery-type problems have also been studied in the computational biology literature, specifically for DNA storage (see e.g.~\cite{organick2018random, DNAS}). In these settings, the population of strings corresponds to a collection of DNA sequences.}

Summarizing, while population recovery from the deletion channel is a natural problem, the above results (and the fact that it is at least as difficult as trace reconstruction) indicate that it is also a hard one.  Thus it is natural to investigate \emph{average-case} versions of this problem; this is the subject of the current work.

\subsection{Our result:  Average-case population recovery in the insertion\hspace{0.04cm}/\hspace{0.04cm}deletion model.}

In the average-case model we consider, there is a given \emph{population size $s\ge 1$}, i.e. there is a set $x^1,\dots,x^s$ of $s$ strings which are assumed to be drawn independently and uniformly from $\{0,1\}^n$.  Associated with this population is an \emph{arbitrary} vector of non-negative probability values $p_1,\dots,p_s$, where $p_i$ is the probability that the distribution ${\cal D}$ puts on string $x^i$. Thus in our model the support of the distribution is ``average-case'' but the actual distribution over that support is ``worst-case.''

Building on the work of \cite{HPP18}, our main result is a highly efficient algorithm for average-case population recovery in the presence of insertions and deletions.  We show that even for extremely large population sizes $s$ (up to $\exp({\Theta(n^{1/3})})$), the average-case population recovery problem can be solved by a highly efficient algorithm which has running time polynomial in $n$ (the length of unknown strings), $s$ (the population size), and $1/\eps$ (where $\eps$ is the total variation distance between $\calD$ and the distribution returned~by the algorithm).  The sample complexity of our
algorithm is polynomial in $s$, $1/\eps$, and $\exp(\log^{1/3} n)$.  Thus our algorithm extends the average-case trace reconstruction results of \cite{HPP18} to the more challenging setting of $s$-string population recovery with essentially the best possible dependence on the new parameters $s$ and $1/\eps$ (which are not present in the original trace reconstruction problem but are inherent in the population recovery problem).

In more detail, we prove the following theorem (the exact definition of a random trace drawn from the insertion\hspace{0.04cm}/\hspace{0.04cm}deletion noise channel ${\cal C}_{q,q'}({\cal D})$ will be given in~\Cref{sec:preliminaries}):

\begin{theorem} \label{thm:main}
Fix any two  constants $q,q'\in [0,1)$ as \emph{deletion} and \emph{insertion rates}, respectively.
There is an algorithm $A$ with the following property:  Let {$\delta_{\text{hard}} \geq \exp( -\Theta(n^{1/3}) )$} be a \emph{fraction of hard support sets}, let {$\delta_{\text{fail}} \geq \exp( -\Theta(n^{1/3}) )$} be a \emph{failure probability}, let $\eps \geq \exp( -\Theta ( n^{1/3} ) )$ be an \emph{accuracy parameter},
let $1 \leq s \leq \exp(\Theta(n^{1/3}))$ be a \emph{support size}, and let $x^1,\dots,x^s$ be a \emph{support set} (viewed as an ordered list of strings in $\zo^n$).
For at least a $(1-\delta_{\text{hard}})$-fraction of all $2^{ns}$ many possible $s$-element support~sets, it is the case that for any probability distribution ${\cal D}$ supported on $\{x^1,\dots,x^s\}$,  given $n,s,\eps,\delta_{\text{hard}}, \delta_{\text{fail}}$, and access to ${\cal C}_{q,q'}({\cal D})$, algorithm $A$ uses $\smash{ \text{\emph{poly}} ( s, 1 / \eps, \exp ( \log^{1/3} n ) , \exp ( \log^{1/3} ( 1 / \delta_{\text{hard}})) , \log(1 / \delta_{\text{fail}}) ) }$ random traces from ${\cal C}_{q,q'}({\cal D})$, runs in time ${ \text{\emph{poly}} ( n, s, 1 / \eps, 1 / \delta_{\text{hard}}, \log(1 / \delta_{\text{fail}}) ) }$ and has the following property:  with probability at least $1- \delta_{\text{fail}}$ it outputs a hypothesis distribution ${\cal D}'$ over $\zo^n$ such that $\dtv({\cal D},{\cal D}') \leq \eps.$
\end{theorem}

\noindent {\bf Discussion.}
Taken together with the recent results of \cite{BCFSS19},~\Cref{thm:main} shows that the average-case and worse-case versions of population recovery in the presence of insertions and deletions have dramatically different complexities. The best known algorithm for the population size-$s$ worst-case population recovery problem \cite{BCFSS19} has a doubly exponential dependence on $s$; even for $s$ constant this sample complexity is significantly worse than the best known sample complexity for the $s=1$ worst-case trace reconstruction problem, which is $\exp(\Theta(n^{1/3}))$ by \cite{DOS17,NazarovPeres17}.  The $n^{\Omega(s)}$ sample complexity lower bound given in \cite{BCFSS19} shows that an exponential dependence on $s$ is inherent for worst-case population recovery.  In contrast, \Cref{thm:main} shows that a \emph{polynomial} sample complexity (and running time) dependence on $s$ is achievable for the average-case problem, and that passing from $s=1$  to larger values of $s$ does not incur much increase in complexity for the average-case problem.

In independent work, \cite{KMMP19} studied a different generalization of trace reconstruction which they called \emph{matrix reconstruction}. Instead of reconstructing a string by sampling traces where each character of the string has some probability of being deleted, the goal in \emph{matrix reconstruction} is to reconstruct a matrix by sampling traces where each row and column of the matrix has some probability of being deleted. They used similar techniques to those used in this paper.

\subsection{Our techniques}

A natural way to approach our problem is to attempt to reduce it to the $s=1$ case, which as described above is just the average-case trace reconstruction problem which was solved by  \cite{HPP18}. However, two challenges arise in carrying out such a reduction.  The first challenge is that the analysis of \cite{HPP18} only gives an algorithm which succeeds on a $1-\Theta(1/n)$ fraction of all source strings $x\in \{0,1\}^n$. So if the population size $s$ is much larger than $n$, then a random population of $s$ source strings will with high probability contain $\Theta(s/n)$ many ``hard-to-reconstruct'' strings. It is not clear how to proceed if these hard-to-reconstruct strings have significant weight under the distribution ${\cal D}$ (which they may, since the distribution ${\cal D}$ over the $s$ source strings is assumed to be completely arbitrary).

We get around this challenge by showing that for any arbitrarily small $\delta$, the algorithm of~\cite{HPP18} can be extended in a black-box way to succeed on any $1-\delta$ fraction of all source strings $x \in \zo^n$ (at the cost of a modest increase in running time and sample complexity depending on the value of $\delta$).  By taking $\delta \ll 1/s$, with high probability the random population will consist entirely of source strings $x^i$ each of which could be reconstructed in isolation
   if we were given only traces coming from $x^i$.

The second (main) challenge, of course, is that we are not given traces from each individual string $x^i$ in isolation, but rather we are given a mixture of traces from all $s$ strings $x^1,\dots,x^s$.  The main contribution of our work is a clustering procedure which lets us (with high probability over a population of $s$ random source strings) correctly group together traces that came from the same source string.  Given the ability to do such clustering, we can indeed use the \cite{HPP18} algorithm on each obtained cluster to identify each of the source strings which has non-negligible weight under the distribution ${\cal D}$, and given the identity of these source strings that each trace came from, it is straightforward to output a high-accuracy hypothesis for the unknown distribution ${\cal D}$ over these strings.

The core clustering procedure which we develop is a simple algorithm which we call $A_{\textrm{cluster}}$ (see~\Cref{fig:clustering} in~\Cref{sec:clustering}).  This algorithm takes two traces $a$ and $b$ as input, and outputs either ``same'' or ``different.'' Its performance guarantee is the following: If $a$ and $b$ were generated as two independent traces from \emph{the same} randomly chosen source string $\bx$, then with high probability $A_{\mathrm{cluster}}$ outputs ``same,'' whereas if $a,b$ were generated as two traces coming from \emph{two independent} uniform random source strings $\bx^1$ and $\bx^2$ respectively, then with high probability $A_{\mathrm{cluster}}$ outputs ``different.''  (See~\Cref{thm:clustering} for a detailed statement.)

The idea underlying $A_{\mathrm{cluster}}$ is as follows.  Given a trace $a$ (which we view as a string over $\{-1,1\}$), imagine breaking it up into contiguous segments (which we call ``\emph{blocks}'') and summing the $\pm 1$ bits within each block, and let $\ssum(a,i)$ denote the sum of the bits in the $i$-th block.  We do the same for the trace $b$ and obtaining a value $\ssum(b,i)$ from the $i$-th block of $b$.  The high-level idea is that,
for a suitable
choice of the block size, in general there will be significant overlap between the positions in $\{1,\ldots,n\}$ (of the source string) that gave rise to the elements of the $i$-th block
  of $a$ and the $i$-th block of $b$.  As a result,
\begin{itemize}
\item On the one hand, if $a$ and $b$ came from the same source $\bx$, then there will be significant cancellation in the difference $\ssum(a,i) - \ssum(b,i)$ and this difference will tend to be ``small'' in magnitude.
\item On the other hand, if $a$ and $b$ came from independent source strings $\bx^1$ and $\bx^2$, then there will be no such cancellation and the difference $\ssum(a,i)-\ssum(b,i)$ will not be so ``small'' in magnitude.
\end{itemize}
Therefore by checking the magnitude of $\ssum(a,i)-\ssum(b,i)$ across many different blocks $i$, it is possible to determine with high confidence whether or not $a$ and $b$ came from the same source string or not.


\section{Preliminaries} \label{sec:preliminaries}

We write $[n]=\{1,\ldots,n\}$ for a positive integer $n$.
We index strings $x \in \zo^n$ as $x=(x_1,\dots,x_n).$ We use {\bf bold font} to denote random variables (which may be real-valued, integer-valued, $\{0,1\}^\ast$-valued, etc.).

We consider an insertion-deletion noise channel ${\cal C}_{q,q'}$ defined as by \cite{HPP18}.  Given a \emph{deletion rate} $q$ and an \emph{insertion rate} $q'$, both in $[0,1)$, the insertion-deletion channel ${\cal C}_{q,q'}$ acts on an $x \in$ $\zo^n$ as follows: First, for each $j\in [n]$, $\bG_j(q')-1$ many independent and uniform bits from~$\zo$ are inserted before the $j$-th bit of $x$, where  $\bG_1(q'),\dots,\bG_n(q')$ are i.i.d.~geometric random variables satisfying
\[
\Pr\big[\bG_j(q')=\ell\big] = (q')^{\ell-1} (1-q')
\]
(i.e.~each $\bG_j(q')$ is distributed as Geometric$(1-q')$). Then each bit of the resulting string is independently deleted with probability $q$. The resulting string is the output from ${\cal C}_{q,q'}(x)$, and we write ``$\by \sim {\cal C}_{q,q'}(x)$'' to indicate that $\by$ is a random trace generated from $x$ in this way. If ${\cal D}$ is a distribution over $n$-bit strings, we write
``$\by \sim {\calC}_{q,q'}({\cal D})$'' to indicate that $\by$ is obtained by first drawing 
   $\bx \sim {\cal D}$ and then drawing $\by \sim {\cal C}_{q,q'}(\bx)$.


\section{Achieving an arbitrarily small fraction of ``hard'' strings in average-case trace reconstruction} \label{sec:reduction}

Fix any constants $q,q'\in [0,1)$ as deletion and insertion rates, receptively.
We will use asymptotic notation such as $O(\cdot)$ and $\Theta(\cdot)$ to
  hide constants that depend on $q$ and $q'$.
  
The main result of \cite{HPP18} is an algorithm which successfully performs trace reconstruction on at least $(1-O(1)/n)$-fraction of all $n$-bit strings (which is 
  $1-M/n$ for some constant $M=M(q,q')$ that only depends on $q$ and $q'$).
In more detail, their main result is the following:

\begin{theorem} \label{thm:HPP}
Fix any constants $q,q'\in [0,1).$
There is a deterministic algorithm $A_{\text{average-case}}$ with the following property:
It is given (1) a confidence parameter $\delta>0$, (2) the length $n$ of
  an unknown string $x\in \{0,1\}^n$ and (3)
  access to $\calC_{q,q'}(x)$,
  uses 
\begin{equation}\label{hehehehe1}
\exp\left(O\left( \log^{1/3}n\right)\right)\cdot \log\left(1/\delta\right)
\end{equation}
  traces 
  drawn from ${\cal C}_{q,q'}(x)$, and runs in time 
  $\poly(n,\log(1/\delta))$.
%
For at least \mbox{$(1 - O(1/n))$}-fraction~of all strings $x \in \zo^n$,\footnote{Theorem~1 of \cite{HPP18} only claims a $1-o_n(1)$ fraction of strings $x$, but the proof shows that the fraction is $1 - O_{q,q'}(1)/n$; see e.g.~the discussion at the beginning of Section~1.3 of \cite{HPP18}.} it is the case that,
  algorithm $A_{\text{average-case}}\hspace{0.04cm}(\delta,n,\calC_{q,q'}(x))$ outputs 
  the string $x$ with probability at least $1-\delta$
  \emph{(}over the randomness of traces drawn from $\calC_{q,q'}(x)$\emph{)}. 
\end{theorem}

Note that in the above theorem the fraction of ``hard'' strings $x \in \zo^n$ on which the \cite{HPP18} algorithm does not succeed is $\Theta(1/n)$. In our setting, to achieve results for general population sizes $s$, we may require the fraction of ``hard'' strings on which the reconstruction algorithm does not succeed to be smaller than this; to see this, suppose for example that we are considering a population of size $s=n^2$.  If a $\Theta(1/n)$ fraction of strings are ``hard'' and $n^2$ strings are chosen uniformly at random to form the support of our distribution ${\cal D}$, then we would expect $\Theta(n)$ many hard strings to be present in the support set (i.e.~the population) of $n^2$ strings. If the unknown distribution over the $n^2$ strings (which, recall, may be any distribution over that support) puts a significant amount of its probability mass on these hard strings, then it may not be possible to successfully recover the population.

In this section we show that the fraction of strings in $\zo^n$ that are ``hard'' can be driven down from $\Theta(1/n)$ to an arbitrarily small fraction in the \cite{HPP18} result, at the cost of a corresponding modest increase in the sample complexity and running time of the algorithm. (As suggested by the discussion given above, such an extension is crucial for us to be able to handle populations of size $s=\omega(n).$)  It may be possible to verify this directly via a careful reworking of the \cite{HPP18} proof, but that proof is involved and such a verification would be quite tedious. Instead we give a simple and direct argument which uses~\Cref{thm:HPP} in a black-box way to prove the following generalization of it, in which only an arbitrarily small fraction of strings  are hard to reconstruct:

\begin{theorem} \label{thm:HPPplus}
Fix any constants $q,q'\in [0,1)$.
There is a deterministic algorithm $A'_{\text{average-case}}$~with the following property:
It is given (1) 
  $\tau>0$ as the desired fraction of hard strings,
  (2) a confidence parameter $\delta$, (3) the length $n$ of the unknown string 
  $x\in \{0,1\}^n$ and (4) access to $\calC_{q,q'}(x)$. It  uses 
\[
\exp\left(O\left( \big(\log\hspace{0.04cm} \max\{n,1/\tau\}\big)^{1/3}\right)\right)\cdot \log\left(1/\delta\right)
\]
many traces drawn from $\calC_{q,q'}(x)$, and runs in time 
  $\poly(\max\{n,1/\tau\},\log(1/\delta)).$ 
For at least $1-\tau$ fraction of all strings $x\in \{0,1\}^n$,
  it is the case that algorithm $A'_{\text{average-case}}\hspace{0.04cm}(\tau,\delta,n,\calC_{q,q'}(x))$ 
  outputs the string $x$ with probability at least $1-\delta$.
   
\end{theorem}

We note that the sample complexity of~\Cref{thm:HPPplus} interpolates smoothly between the average-case result of \cite{HPP18}, in which a $\tau = \Theta(1/n)$ fraction of strings are hard, and the worst-case results of  \cite{DOS17,NazarovPeres17}, in which no strings in $\zo^n$ (equivalently, at most a $\tau = 1/2^{n+1}$ fraction of strings) are hard.

The high-level idea underlying~\Cref{thm:HPPplus} is very simple: By padding the input string $x$ (which should be thought of as uniformly random over $\zo^n$) with random bits, it is possible to obtain a uniformly random $N$-bit string, and by running algorithm $A_{\text{average-case}}$ over this string of length $N$, with probability $1-\Theta(1/N)$ it is possible to reconstruct this $N$-bit string, from which the original input string $x$ can be reconstructed. Taking $N$ to be suitably large this yields the desired result. We give a detailed proof below.

Let $M$ be the constant hidden in the $O(1/n)$ in Theorem \ref{thm:HPP}.
We~note that if $\tau \geq M/n$ then we may simply use $A_{\text{average-case}}$, so we henceforth assume that $\tau < M/n.$

\medskip
\noindent {\bf The algorithm.}  Algorithm $A'_{\text{average-case}}
(\tau,\delta,n,\calC_{q,q'}(x))$ works by running an auxiliary algorithm 
$A^*(\tau,n,\calC_{q,q'}(x))$ (which always outputs an $n$-bit string) $O(\log(1/\delta))$ many times.  If at least $9/16$ of the $O(\log(1/\delta))$ runs of $A^*$ yield the same $n$-bit string then this is the output of $A'_{\text{average-case}}$, and otherwise $A'_{\text{average-case}}$ outputs ``failure.''  Below we will show that for at   least $1-\tau$ fraction of all strings $x\in \{0,1\}^n$, 
$A^*(\tau,n,\calC_{q,q'}(x))$ outputs the correct string $x$ with 
probability at least $5/8$.
It follows from the Chernoff bound that $A'_{\text{average-case}}$
achieves the desired $1-\delta$ success probability.

We turn to describing and analyzing $A^*\hspace{0.04cm}(\tau,n,\calC_{q,q'}(x))$, which  works as follows:
\begin{flushleft}
	\begin{enumerate}
		\item It draws a string $\bz$ uniformly from $\zo^{N-n}$ (the value of $N>n$ will be specified later).
		
		\item Let $m=m(N)$ be the following parameter:
		\[m(N)=\exp\left( O\left(\log^{1/3} N\right)\right)\cdot \log\left(1/\delta'\right),\] where 
		$\delta' = 1/8$. This is the number of traces needed by
		$A_{\text{average-case}}$ to achieve confidence parameter $\delta'$ on
		strings of length $N$ (as in (\ref{hehehehe1})). 
		For $m$ times, algorithm $A^*$ independently repeats the following:  at the $i$-th repetition it draws a string $\by^{(i)} \sim {\cal C}_{q,q'}(x)$, constructs a string $\by'^{(i)}$ that is distributed according to ${\cal C}_{q,q'}(\bz)$, and constructs $\ba^{(i)} := \by^{(i)} \circ \by'^{(i)}$ which is the concatenation of $\by^{(i)}$ and $\by'^{(i)}.$
		
		\item Finally, it uses the $m$ strings $\ba^{(1)},\dots,\ba^{(m)}$ to run algorithm $A_{\text{average-case}}\hspace{0.04cm}$ with length $N$ and 
		confidence parameter $\delta'$.  
		Let $w \in \zo^{N}$ be the  string that $A_{\text{average-case}}$ returns.  The output of $A^*$ is $w_1 w_2 \cdots w_n$, the first $n$ characters of $w$.
	\end{enumerate}
\end{flushleft}

\noindent
{\bf Proof of correctness.}  We first observe that (as an immediate consequence of the definition of the noise channel ${\cal C}_{q,q'}$) each string $\ba^{(i)} = \by^{(i)} \circ \by'^{(i)}$ generated as in Step~2 of $A^*\hspace{0.04cm}(\tau,n,\calC_{q,q'}(x))$  is distributed precisely as a draw from ${\cal C}_{q,q'}(x \circ \bz).$ By the choice of $m=m(N)$ in Step~2, the strings $\ba^{(1)},\dots,\ba^{(m)}$ constitute precisely the required traces for a run of $A_{\text{average-case}}$ on the $N$-bit string $x \circ \bz$.

Let us say that the strings in $\zo^N$ which $A_{\text{average-case}}$ 
(with parameters $\delta'$ and $N$) correctly reconstructs with probability at least $1-\delta'$ are \emph{good} strings, and that the other strings in $\zo^N$ are \emph{bad} strings.  By~\Cref{thm:HPP}, at most an ($M/N$)-fraction of all strings in $\zo^N$ are bad.  The value of $N$ is set to $N := 4M/\tau$,\footnote{Note that $N \geq 4n$ using $\tau<M/N$, so $N-n>0$ and indeed Step~1 makes sense.} so $M/N = \tau/4$, and it is the case that at most a $\tau/4$ fraction of strings in $\zo^N$ are bad.
For each $x \in \zo^n$, let $\gamma_x$ denote the fraction of strings $z \in \zo^{N-n}$ such that $x \circ z$ is bad. The average over all $x \in \zo^n$ of $\gamma_x$ is at most $\tau/4$, and consequently at most a $\tau$ fraction of strings $x$ have $\gamma_x \geq 1/4.$  

\begin{claim} \label{claim:a}
	If $x \in \zo^n$ has $\gamma_x < 1/4$, then $A^*\hspace{0.04cm}(\tau,n,\calC_{q,q'}(x))$ outputs $x$ with probability at least $5/8.$
\end{claim}
\begin{proof}
	The probability that $\bz \sim \zo^{N-n}$ such that $x \circ \bz$ is good is at least $3/4$.  If $x \circ \bz$ is good then with probability at least $1-\delta'=7/8$ the output of $A_{\text{average-case}}$ as run in Step~3 is the string $x \circ \bz$ and hence the output of $A^*$ is $x$. The claim follows since $(3/4) \cdot (7/8) > 5/8.$
\end{proof}

Hence for at least a $(1-\tau)$-fraction of $x \in \zo^n$, a run of $A^*\hspace{0.04cm}(\tau,n,\calC_{q,q'}(x))$ outputs $x$ with probability at least $5/8.$  For any such $x$, a simple Chernoff bound shows that with probability at least $1-\delta$, at least $9/16$ of the $O(\log(1/\delta))$ many independent runs of $A^*$ will output $x$. This concludes the proof of~\Cref{thm:HPPplus}.~\hfill $\blacksquare$


\section{The core clustering result} \label{sec:clustering}

In this section we state and prove the key clustering result that is used
  in the main algorithm.  Intuitively,~it gives an efficient procedure with the following performance guarantee:  Given two traces, the procedure can determine with high probability whether
 the two traces were both obtained as traces from the same uniform random string $\bx \sim \zo^n$, or the two traces were obtained from two independent uniform random strings $\bx^1,\bx^2 \sim \zo^n.$

In more detail, the main result of this section is the following theorem:

\begin{theorem} \label{thm:clustering}
Fix any constants $q,q' \in [0,1)$.
There is a deterministic algorithm $A_{\text{cluster}}$ with the following performance guarantee:
It is given a positive integer $n$ and a pair of binary strings $z$ and $z'$.
Let $\delta_{\text{cluster}} := \exp(-\Theta(n^{1/3}))$. 
Then $A_{\text{cluster}}\hspace{0.04cm}(n,z,z')$ runs in time $O(n)$ and satisfies
  the following two properties:
\begin{flushleft}
\begin{enumerate}
\item Suppose that $\bx$ is uniform random over $\zo^n$ and
  $\bz,\bz'$ are independent draws from ${\cal C}_{q,q'}(\bx)$. Then with probability at least $1 - \delta_{\text{cluster}}$, algorithm $A_{\text{cluster}}\hspace{0.04cm}
  (n,\bz,\bz')$ outputs ``same.''\vspace{-0.06cm}

\item Suppose that $\bx^1,\bx^2$ 
are independent uniform random strings over $\zo^n$,
    $\bz\sim {\cal C}_{q,q'}(\bx^1)$ and $\bz'\sim{\cal C}_{q,q'}(\bx^2)$. Then with probability at least $1-\delta_{\text{cluster}}$, $A_{\text{cluster}}\hspace{0.04cm}(n,\bz,\bz')$ outputs ``different.''\vspace{0.3cm}
\end{enumerate}\end{flushleft}
\end{theorem}

\subsection{Proof of ~\Cref{thm:clustering}}


For convenience, we consider strings over $\{-1,1\}$ instead of $\{0,1\}$
in the rest of this section.
We need the following technical lemma:

\begin{lemma} \label{lem_bernoulli_deviation}
	Let $\tau\in (0,1]$ be a constant.
	Then there exist three positive constants $c_1$, $c_2$ and $c_3$
	(that~only depend on $\tau$) such that
	the following property holds.
	For all positive integers $m$ and $m'$ such that
	$m'\le (1-\tau) m$ and $m$ is sufficiently large,
	letting $\XX_1,\ldots,\XX_m$ be independent and
	uniform  random variables over $\{-1,1\}$, we have
	\[
		\Pr\Big[\big|\XX_1+\cdots+\XX_m\big|\ge c_1\sqrt{m}\Big]\ge c_2+c_3 \quad\text{and}\quad
		\Pr\Big[\big|\XX_1+\cdots+\XX_{m'}\big|\ge c_1\sqrt{m}\Big]\le c_2-c_3.
	\]
\end{lemma}

\begin{proof}
	The Berry-Esseen theorem (see e.g.~\cite{Feller}) establishes closeness between the cdf of a sum of ``well-behaved'' independent random variables (such as $\bX_1,\dots,\bX_m$)  and the cdf of a Normal distribution with the same mean and variance.  By the Berry-Esseen theorem, the probability of $|\XX_1+\cdots+\XX_{m}|\ge c_1\sqrt{m}$ is within an additive $\pm o_m(1)$ of the corresponding probability of $|\bG_1|\ge c_1\sqrt{m}$,
	where 
	$\bG_1 \sim {\cal N}(0,m)$. 

	We first consider the case that $m'$ is not too small compared to $m$, say $m' > m^{1/3}$. In this case the Berry-Esseen theorem implies that the probability of $|\XX_1+\cdots+\XX_{m'}|\ge c_1\sqrt{m}$ is also within an additive $\pm o_m(1)$ of the corresponding probability for Gaussian random variables, which is now $\Pr[\hspace{0.04cm} |\bG_2 |\ge c_1\sqrt{m}\hspace{0.04cm}]$ with $\bG_2 \sim {\cal N}(0,m')$. So in this case Lemma~\ref{lem_bernoulli_deviation} is an immediate consequence of an analogous statement for Gaussian random variables,
	\begin{equation}
		\Pr\Big[\big|\bG_1\big|\ge c_1\sqrt{m}\Big]\ge c_2+c_3 \quad\text{and}\quad
		\Pr\Big[\big|\bG_2\big|\ge c_1\sqrt{m}\Big]\le c_2-c_3, \label{eq:gaussians}
	\end{equation}
	where $\bG_1 \sim {\cal N}(0,m)$ and $\bG_2 \sim {\cal N}(0,m')$. The first probability  in (\ref{eq:gaussians}) is the probability that a Gaussian's magnitude exceeds its mean by at least $c_1$ standard deviations, while the second probability in (\ref{eq:gaussians}) is the probability that a Gaussian's magnitude exceeds its mean by at least $c_1/\sqrt{1-\tau}$ standard deviations.  Given this, for suitable $c_1,c_2,c_3$ depending only on $\tau$, the inequalities (\ref{eq:gaussians}) are a straightforward consequence of the following standard bounds on the cdf of a Gaussian $\bG\sim \calN(0,\sigma^2)$
	[\cite{Feller}, Section 7.1]:
	$$
	\left(\frac{1}{x}-\frac{1}{x^3}\right)\cdot \frac{e^{-x^2/2}}{\sqrt{2\pi}}\le 
	\Pr\big[\bG\ge x\sigma\big]\le \frac{1}{x}\cdot \frac{e^{-x^2/2}}{\sqrt{2\pi}},\quad
	\text{for all $x>0$.}
	$$
	
	Finally we consider the case that $m'$ is very small compared to $m$, say $m' \leq m^{1/3}$.  In this case, by~the Berry-Esseen theorem we have that $\smash{\Pr[\hspace{0.04cm} |\XX_1+\cdots+\XX_{m}|\ge c_1 \sqrt{m}\hspace{0.04cm}]}$ is $\pm o_m(1)$-close to the probability that a Gaussian's magnitude exceeds its mean by $c_1$ standard deviations, which is at least some absolute constant, while 
	$\smash{\Pr[\hspace{0.04cm} |\XX_1+\cdots+\XX_{m'}|\ge c_1 \sqrt{m}\hspace{0.04cm}]}$ is zero for sufficiently large $m$, because $\smash{m' \leq m^{1/3} < c_1 \sqrt{m}}$ for $m$ sufficiently large. This finishes the proof of Lemma \ref{lem_bernoulli_deviation}.
\end{proof}

Recall that constants $q,q' \in [0,1)$ denote the deletion probability and insertion probability respectively. Let $p,p'\in (0,1]$ be $p=1-q$ and $p'=1-q'$.
Then the expected length of a string drawn from $\calC_{q,q'}(x)$ with
$x\in \{-1,1\}^n$
is $np/p'=\alpha n$, where $\alpha:=p/p'$ is a positive constant.

Given a string $x\in \{-1,1\}^n$,
we start by describing an equivalent way of drawing
$\bz\sim \calC_{q,q'}(x)$. 
We say a string $r$ over $[n]\cup \{*\}$
is an $n$-\emph{pattern} if every $i\in [n]$ appears
in $r$ at most once and integers appear~in~$r$ in ascending order.
We write $\calR_{n,q,q'}$ to denote the following distribution
over $n$-patterns.
To draw $\rr\sim \calR_{n,q,q'}$ we start with $r^{(0)}=(1,2,\ldots,n)$.
Then for each $j\in [n]$, $\bG_j(q')-1$ many $*$'s are inserted
before the $j$-th entry (with value $j$) of $r^{(0)}$ to obtain $\rr^{(1)}$.
Finally each entry of $\rr^{(1)}$ is independently deleted with probability $q$
to obtain the final string $\rr$.
Using $\calR_{n,q,q'}$, drawing $\bz\sim \calC_{q,q'}(x)$
can be done equivalently as follows:
\begin{enumerate}
	\item Draw an $n$-pattern $\rr\sim\calR_{n,q,q'}$.\vspace{-0.12cm}
	\item For each index $i\in [n]$ that appears in $\rr$, replace it by $x_i$ in $\rr$.\vspace{-0.12cm}
	\item Replace each $*$ in $\rr$ with an independent and uniform draw from $\{-1,1\}$.
\end{enumerate}

Next we introduce a number of parameters and constants that will
be used in the
clustering algorithm $A_{\text{cluster}}$. 
Two parameters $\tilde{s}$ and $t$ used in the algorithm are
\[
t:=n^{2/3}\quad \text{and}\quad \tilde{s}=\left\lfloor\frac{\alpha n}{4t}\right\rfloor=\Theta(n^{1/3})
\]
so that $2\tilde{s}t\le \alpha n /2$.
Three constants $\beta,\gamma$ and $\delta$ are defined
as $c_1,c_2$ and $c_3$ in Lemma \ref{lem_bernoulli_deviation} with
$\tau$ set to be the following constant in $(0,1)$: 
$
	\tau= 0.7\hspace{0.02cm}p\hspace{0.02cm}p'$.
For each $\ell\in [\tilde{s}]$, we let $I_\ell$ denote the following set of integers:
\begin{equation} \label{eq:Iell}
I_\ell=\big[(2\ell-2)t+1, (2\ell-1)t\big] \cap \Z.
\end{equation}
Given a string $z$ over $\{-1,1\}$ (or an $n$-pattern $r$),
we will refer to entries $z_i$ of $z$ (or $r_i$ of $r$) over $i\in I_\ell$
as the $\ell$-th \emph{block} of $z$ (or $r$).
So each block consists of $t$ entries and
two consecutive blocks are separated by a gap of $t$ entries.
Given an $n$-pattern $r$ and an integer $\ell$, we write
$B_\ell(r)\subset [n]$ to denote the set of $i\in [n]$ that appears in
the $\ell$-th block of $r$.

The algorithm $\smash{A_{\text{cluster}}}$
is described in Figure~\ref{fig:clustering}.
Before stating the key technical~lemma (Lemma \ref{lem:keytechnical})
and using it to prove Theorem \ref{thm:clustering}, we give some intuition for the algorithm
$A_{\text{cluster}}$.

\begin{figure}[h]
	\begin{framed}
		\begin{minipage}{42em}
			\begin{flushleft}
				\noindent \textbf{Algorithm} $A_{\text{cluster}}\hspace{0.04cm}(n,z,z')$\\
				\noindent {\bf Input:} A positive integer $n$ and two strings $z$ and $z'$ over $\{-1,1\}$.

				{\bf Output:}  ``Same'' or ``different.'' \vspace{0.0cm}
				\begin{enumerate}
					\item For each $\ell\in [\tilde{s}]$, set $Z_\ell$
					to be the sum of $z_i$ over $i\in I_\ell$, with $z_i=0$ when $i>|z|$.\vspace{-0.05cm}
					\item For each $\ell\in [\tilde{s}]$, set $Z_\ell'$
					to be the sum of $z_i'$ over $i\in I_\ell$, with $z_i'=0$ when $i>|z'|$.\vspace{-0.05cm}
					\item Count the number of $\ell\in [\tilde{s}]$ such that
					$\big|Z_\ell-Z_\ell'\big|\ge \beta \sqrt{2t}$.\vspace{-0.05cm}
					\item If the number of such $\ell$ is at least $\gamma \tilde{s}$, return ``different;''
					otherwise, return ``same.''  \vspace{-0.2cm}

				\end{enumerate}
			\end{flushleft}
		\end{minipage}\vspace{0.25cm}
	\end{framed}\vspace{-0.2cm}
	\caption{Description of the clustering algorithm $A_{\text{cluster}}$.} \label{fig:clustering}
\end{figure}

Recall the two cases in Theorem \ref{thm:clustering}.
We start with the easier second case, where
$\bx^1$ and $\bx^2$ are drawn from $\{-1,1\}^n$ uniformly and independently,
$\bz\sim \calC_{q,q'}(\bx^1)$
and $\bz'\sim \calC_{q,q'}(\bx^2)$.
First it is easy to show (see property (0) of Lemma \ref{lem:keytechnical}) that
$|\bz|,|\bz'|\ge \alpha n/2$ with very high probability.
When this happens,
$\bZ_\ell$ is the sum of $t$ independent and uniform random variables over $\{-1,1\}$
and the same holds for $\bZ_\ell'$. Moreover, $\bZ_\ell$ and $\bZ_\ell'$ are independent of each other since $\bx^1$ and $\bx^2$ are drawn independently and thus,
$\bZ_\ell-\bZ_\ell'$ can be equivalently
written as the sum of $2t$ independent and uniform variables over $\{-1,1\}$.
Furthermore, the $\tilde{s}$ random variables $\bZ_\ell-\bZ_\ell'$ over $\ell\in [\tilde{s}]$
are independent.
Thus, it follows from Lemma \ref{lem_bernoulli_deviation} and our choices of $\beta,\gamma$ and $\delta$
that the probability of each $|\bZ_\ell-\bZ_\ell'|\ge \beta \sqrt{2t}$
is at least
$\gamma+\delta$ and with very high probability,
the number of such $\ell\in [\tilde{s}]$ is at least $\gamma \tilde{s}$,
in which case the algorithm returns ``different'' as desired.

In the first case of Theorem \ref{thm:clustering}, we draw $\bx$ from $\{-1,1\}^n$
  uniformly and
then draw $\bz,\bz'$ from $\calC_{q,q'}(\bx)$ independently.
Equivalently one can view the process as first drawing two $n$-patterns
$\rr$ and $\rr'$ independently from $\calR_{n,q,q'}$
and $\bx$ from $\{-1,1\}^n$.
The string $\bz$ (or $\bz'$) is then obtained by replacing each $i\in [n]$
in $\rr$ (or $\rr'$) by $\bx_i$ and each $*$ by an independent draw from $\{-1,1\}$.
Again we assume that $|\rr|,|\rr'|\ge \alpha n/2$, which happens with high probability.
When this is the case,
each of $\bZ_\ell$ and $\bZ_\ell'$ for $\ell\in [\tilde{s}]$ remains the sum
of $t$ independent 
uniform random variables over $\{-1,1\}$.
However, when an index $i\in [n]$ appears in the $\ell$-th block
of both $\rr$, $\rr'$, then $\bx_i$ appears in 
both sums $\bZ_\ell, \bZ'_\ell$
and gets cancelled out in their difference $\bZ_\ell-\bZ_\ell'$.

Our main technical lemma shows that with very high probability
over draws $\rr$ and $\rr'$ from $\calR_{n,q,q'}$, the
following two properties hold: (1) $|B_\ell(\rr)\cap B_\ell(\rr')|\ge \tau t$ for
every $\ell\in [\tilde{s}]$, i.e., there are at least $\tau t$ many integers
that appear in the $\ell$-th block of both $\rr$ and $\rr'$; and (2)
No index $i\in [n]$ appears in two different blocks of $\rr$ and $\rr'$
(i.e., it cannot be the case that both $i\in B_\ell(\rr)$ and $i\in B_{\ell'}(\rr')$
with $\ell\ne \ell'$; intuitively the reason why we leave
a gap of $t$ entries between two consecutive blocks is to achieve this property).
Fixing such a pair of $n$-patterns $r$ and $r'$,
property (1) implies that each $\bZ_\ell-\bZ_\ell'$ can be written
as the sum of at most $(1-\tau)2t$ many independent $\{-1,1\}$-variables; given this,
it follows directly
from Lemma \ref{lem_bernoulli_deviation} that the probability of
each $|\bZ_\ell-\bZ_\ell'|
\ge \beta\sqrt{2t}$ is at most $\gamma -\delta$.
Furthermore (2) implies that
the $\tilde{s}$ variables $\bZ_\ell-\bZ_\ell'$ 
over $\ell\in [\tilde{s}]$ are
independent. This lets us easily infer that the number of $\ell$ such that
$|\bZ_\ell-\bZ_\ell'|\ge \beta\sqrt{2t}$ is less than $\gamma \tilde{s}$
with very high probability, in which case the algorithm returns ``same'' as desired.

As discussed above, the main technical lemma we require is as follows:

\begin{lemma}\label{lem:keytechnical}
	Let $\rr,\rr'$ be two $n$-patterns drawn independently from $\calR_{n,q,q'}$.
	Then with probability at least $1-\exp(-\Omega(n^{1/3}))$, the
	following three properties all hold:
	\begin{flushleft}\begin{enumerate}
			\item[\emph{(0):}] $|\rr|,|\rr'|\ge \alpha n/2$.\vspace{-0.08cm}
			\item[\emph{(1):}] $\big|B_\ell(\rr)\cap B_\ell(\rr')\big|\ge \tau t$ \hspace{0.05cm}for all $\ell\in [\tilde{s}]$.\vspace{-0.08cm}
			\item[\emph{(2):}] If an $i\in [n]$ appears in both
			$B_\ell(\rr)$ and $B_{\ell'}(\rr')$ for some $\ell,\ell'\in [\tilde{s}]$,
			then we have $\ell=\ell'$.
	\end{enumerate}\end{flushleft}
\end{lemma}

The detailed proof of Lemma \ref{lem:keytechnical} is given in~\Cref{sec:keytechnical}; here we give some intuition.

To prove Lemma \ref{lem:keytechnical}, we show
  that $\rr,\rr'\sim \calR_{n,q,q'}$ satisfy each of the three properties with probability
  at least $1-\exp(-\Omega(n^{1/3}))$; the lemma follows from a union bound.
Property (0) follows from tail bounds on sums of independent Geometric random variables and from standard Chernoff bounds (see
  Claim \ref{clm_insertion}) for insertions and deletions,
  respectively.
Indeed property (0) holds with probability $1-\exp(-\Omega(n))$.

Properties (1) and (2) follow from Lemma
  \ref{lem_coordinate_deviation_t} in \Cref{sec:keytechnical}.
To state the lemma, recall that we write $\rr^{(1)}$ to denote the string
  over $[n]\cup \{*\}$ obtained after insertions during the
  generation of $\rr\sim \calR_{n,q,q'}$.
For each $i\in [n]$, we use $\bY_i$ to denote the number of characters
  before $i$ in $\rr^{(1)}$ that survive deletions;
  note that $\bY_i$ is the number of characters that appear before $i$
  in $\rr$ if $i$ survives in $\rr$, but $\bY_i$ is well defined even if $i$ was deleted.
By definition, we have $\smash{\bE [\bY_i ]=((i/p')-1)p}$.
Lemma \ref{lem_coordinate_deviation_t} shows that for constant $c \in (0,1)$, with
  probability $1-\exp(-\Omega(n^{1/3}))$,
  $|\bY_i-\bE[\bY_i]|\le ct$. 
Lemma \ref{lem_coordinate_deviation_t} again follows from tail bounds on sums of independent Geometric random variables and from standard Chernoff bounds.
We define $\bY_i'$ similarly for $\rr'$ and the same statement also holds for $\bY_i'$.

Property (2) follows directly from Lemma \ref{lem_coordinate_deviation_t},
  since for an $i\in [n]$
  to appear in two different~blocks, it must be the case that
  $|\bY_i-\bY_i'|\ge t$ and thus, either $|\bY_i-\bE[\bY_i]|\ge t/2$
  or $|\bY_i'-\bE[\bY_i']|\ge t/2$ (as we have $\bE [\bY'_i] = \bE [\bY_i]$), which happens with probability
  at most $\exp(-\Omega(n^{1/3}))$ by Lemma \ref{lem_coordinate_deviation_t}.

To prove property (1) for $\ell\in [\tilde{s}]$, we focus on the
  following interval of indices in $[n]$:
\[
I_\ell^{(0)} := \left[\frac{p'}{p}(2\ell-1.9)t+1,\frac{p'}{p}(2\ell-1.1)t \right] \cap \Z, 
\]
and show that with probability at least $1-\exp(-\Omega(n^{1/3}))$, we have both
\begin{enumerate}[label=(\alph*)]
	\item At least $\tau t=0.7t p p'$ indices in~$I_\ell^{(0)}$ survive in both $\rr$ and $\rr'$; and \vspace{-0.07cm}
	\item Every $\smash{i \in I_\ell^{(0)}}$ that survives in both $\rr$ and $\rr'$ lies
	in both $B_\ell(\rr)$ and $B_\ell(\rr')$.
\end{enumerate}
Item (a) follows from a Chernoff bound: the length of $\smash{I_\ell^{(0)}}$
  is $0.8tp'/p$ and every element survives independently in both strings
  with probability $p^2$.
Letting $i_0,i_1$ be the left and right ends of $\smash{I_\ell^{(0)}}$,
  item (b) holds when $\bY_{i_0},\bY_{i_0}',\bY_{i_1},\bY_{i_1}'$ do not
  shift too far ($0.1t$) away from their expectations which happens with probability
  at least $\smash{1-\exp(-\Omega(n^{1/3}))}$ by Lemma \ref{lem_coordinate_deviation_t}.
\medskip

Finally we use Lemma \ref{lem:keytechnical} to prove Theorem \ref{thm:clustering}:\medskip

\begin{proof}[Proof of~\Cref{thm:clustering}]
We start with the second case in which
  $\bx^1,\bx^2$ are independent uniform random strings over $\zo^n$, $\bz \sim {\cal C}_{q,q'}(\bx^1)$ and $\bz' \sim {\cal C}_{q,q'}(\bx^2)$.
By our discussion earlier, $\bz$ and $\bz'$ can be generated equivalently by first
  drawing $\br,\br'\sim \calR_{n,q,q'}$, then drawing $\bx^1,\bx^2$, and finally deriving
  $\bz$ (or $\bz'$) from $\br$~(or $\br'$) using $\bx^1$ (or $\bx^2$) as well as
  independent random bits for the $*$'s.
By Lemma \ref{lem:keytechnical}, $\br$ and $\br'$ satisfy all three properties
  with probability at least $1-\exp(-\Omega(n^{1/3}))$.
Fixing $r$ and $r'$ that satisfy all three~properties~(for the first case
  we only need property (0)),
  we show that $A_{\text{cluster}}\hspace{0.04cm}(n,\bz,\bz')$ returns ``different'' with probability
  at least $1-\exp(-\Omega(n^{1/3}))$ conditioning
  on $\br=r$ and $\br'=r'$; the lemma for this case then follows.

To this end, it follows from property (0) that each $\bZ_\ell-\bZ'_\ell$
  is a sum of $2t$ independent uniform random variables over $\{-1,1\}$ and thus,
  each $\ell\in [\tilde{s}]$ satisfies $|\bZ_\ell-\bZ'_\ell|\ge \beta\sqrt{2t}$
  with probability at least $\gamma+\delta$.
Moreover, the $\tilde{s}$ variables $\bZ_\ell-\bZ'_\ell$ are independent.
It follows from a Chernoff bound (and that $\delta$ is a positive constant)
  that $A_{\text{cluster}}$ returns ``different'' with probability $1-\exp(-\Omega(\tilde{s}))
  =1-\exp(-\Omega(n^{1/3}))$.

For the second case we can similarly generate $\bz,\bz'$ by first drawing
  $\rr,\rr'\sim\calR_{n,q,q'}$, then drawing~$\bx$,~and finally
  deriving $\bz,\bz'$ from $\rr,\rr'$ using the same $\bx$ and
  independent random bits for the $*$'s.
Again it follows from Lemma \ref{lem:keytechnical} that $\rr,\rr'$ satisfy
  all three properties with probability $1-\exp(-\Omega(n^{1/3}))$.
Fixing $r,r'$ that satisfy all three properties, we show that $A_{\text{cluster}}
  \hspace{0.04cm}(n,\bz,\bz')$ returns ``same'' with probability
  $1-\exp(-\Omega(n^{1/3}))$, conditioning on $\br=r$ and $\br'=r'$;
  the lemma for this case then follows.

For this purpose, properties (0) and (1) imply that
  each $\bZ_\ell-\bZ_\ell'$ is the sum of at most $(1-\tau)2t$~many independent uniform random variables over $\{-1,1\}$.
Lemma \ref{lem_bernoulli_deviation} implies that the probability
  of $|\bZ_\ell-\bZ_\ell'|\ge$ $\beta\sqrt{2t}$ is at most $\gamma-\delta$.
Moreover, property (2) implies that these $\tilde{s}$ variables $\bZ_\ell-\bZ_\ell'$
  are independent.
It similarly follows from a Chernoff bound that $A_{\text{cluster}}$
  returns ``same'' with probability $1-\exp(-\Omega(\tilde{s}))=1-\exp(-\Omega(n^{1/3}))$.
\end{proof}



\section{Putting the pieces together:  Proof of~\Cref{thm:main}} \label{sec:main-result}

In this section we combine the main results from earlier sections,~\Cref{thm:HPPplus} from~\Cref{sec:reduction} and~\Cref{thm:clustering} from~\Cref{sec:clustering}, together with standard results on learning discrete distributions, to prove~\Cref{thm:main}.

\subsection{Learning discrete distributions}

We recall the following folklore result on learning a discrete distribution from independent samples:

\begin{theorem} \label{thm:learn-discrete}
Fix $\gamma, \kappa > 0, N \in \N$. Let ${\cal P}$ be an unknown probability distribution over the discrete set $\{1,\dots,N\}$, and let $\bS = \{\bi_1,\dots,\bi_m\}$ be independent draws from ${\cal P}$, where $\smash{m=O ({ (N/ {\kappa^2})} \cdot \log(1/\gamma) )}.$ Let ${\hat{\cal P}}_{\bS}$ denote the empirical probability distribution over $[N]$ corresponding to $\bS$.  Then 
with probability at least $1-\gamma$ over the draw of $\bS$, the variation distance $\dtv(\hat{\cal P}_{\bS},{\cal P})$ is at most $\kappa.$
\end{theorem}

We will need a corollary which says that removing low-frequency elements has only a negligible effect:

\begin{corollary} \label{cor:learn-discrete}
Let ${\cal P},m$ and $\bS$ be as above. Let $\bS'$ be the subset of $\bS$ obtained by removing each \mbox{element} $j$ whose frequency in $\bS$ is at most $\kappa/(2N)$, and let ${\hat{\cal P}}_{\bS'}$ denote the empirical distribution over $[N]$ corresponding to $\bS'$.  Then with probability at least $1-\gamma$ over the draw of $\bS'$, $\dtv(\hat{\cal P}_{\bS'},{\cal P})$ is at most $\kappa.$
\end{corollary}

\begin{proof}
By~\Cref{thm:learn-discrete}, with probability at least $1-\delta$ the hypothesis  $\hat{\cal P}_{\bS}$ from~\Cref{thm:learn-discrete} is $\kappa/2$-close to ${\cal P}$. The corollary follows since the variation distance between $\hat{\cal P}_{\bS}$ and $\hat{\cal P}_{\bS'}$ is at most $N \cdot \kappa/(2N)=\kappa/2.$\vspace{-0.35cm}
\end{proof}

\subsection{Proof of~\Cref{thm:main}}

\noindent Algorithm $A$ is given in~\Cref{fig:main}. Its proof of correctness is given below.

\medskip

\begin{figure}[t!]
	\begin{framed}
		\begin{minipage}{42em}
			\begin{flushleft}
				\noindent \textbf{Algorithm} {$A\hspace{0.04cm}(n,s,\eps,\delta_{\text{hard}},
					\delta_{\text{fail}},\calC_{q,q'}(\calD))$, with constants $q,q'\in [0,1)$
					as deletion and insertion rates.} \\
				\noindent {\bf Input:} String length $n$, support size $\smash{s\le \exp(\Theta(n^{1/3}))}$, accuracy parameter $\smash{\eps \geq \exp (-\Theta (n^{1/3}))}$, 
				fraction of hard support sets   
				$\smash{\delta_{\text{hard}} \geq \exp (-\Theta({n^{1/3}}))}$,
				failure probability ${\smash{\delta_{\text{fail}} 
						\geq \exp(-\Theta(n^{1/3}))}}$,
				and access to 
				$\smash{{\cal C}_{q,q'}({\cal D})}$ where $\calD$ is a  
				probability distribution over $s$ strings in $\{0,1\}^n$.
				
				\noindent {\bf Output:}  Either a probability distribution $\mathcal{D}'$ or ``fail.'' \vspace{0.0cm}
				\begin{enumerate}[itemsep=+0.6mm]
					\item Draw $T$ traces $\by^1,\ldots,\by^T$ from ${\cal C}_{q,q'}({\cal D})$, where 
					$$T = \frac{s}{\eps^2} \cdot \exp\left(\Theta\left(\left(\log   \max \left\{ n, \frac{2s}{\delta_{\text{hard}}} \right\}\right)^{1/3}\right)\right) \cdot \log \left( \frac{3s}{\delta_{\text{fail}}} \right).\vspace{-0.08cm}$$
					
					\item For each pair of traces $\by^i,\by^j$ with $1 \leq i < j \leq T$, run $A_{\text{cluster}}\hspace{0.04cm}(n,\by^i,\by^j)$ 
					from~\Cref{sec:clustering}.
					If the $\smash{{T \choose 2}}$-many outcomes of $A_{\text{cluster}}$ (corresponding to $\smash{{T \choose 2}}$ many answers of ``same'' or ``different'') do not correspond to a disjoint union of cliques then halt and output ``fail,'' otherwise continue.
					
					\item Let the resulting clusters\hspace{0.05cm}/\hspace{0.05cm}cliques be denoted $C_1,\dots,C_r$, so $C_1 \sqcup \cdots \sqcup C_r$ is a partition of the\\ set $\smash{\{\by^1,\dots,\by^T\}}$ of traces.\footnote{Strictly speaking, each $C_i$ is a multiset.} Call $C_i$ \emph{large} if it contains at least $T \cdot (\eps/(2s))$ many elements. \\ Let $C'_1,\dots,C'_{r'}$ denote the large clusters for some $r'\le r$, and let $C'_{\text{total}} = \sum_i |C'_i|$. 
					
					\item For each large multiset $C'_i$, 
					run $A'_{\text{average-case}}$ from~\Cref{sec:reduction}
					using $n$ and strings from $C'_i$, in which\\ $\tau$
					is set to $\smash{\delta_{\text{hard}}/{(2s)}}$ 
					and $\delta$ is set to $\smash{\delta_{\text{fail}}/(3s)}$. 
					Let $\smash{z^{i}}$ be the output of $\smash{A'_{\text{average-case}}}$ on this input.
					
					\item Distribution ${\cal D}'$ that $A$ outputs is supported on $z^{1},\dots,z^{r'}$ and puts weight ${|C'_i|}/{C'_{\text{total}}}$ on $z^{i}$. \vspace{0.1cm}
				\end{enumerate}
			\end{flushleft}
		\end{minipage}
	\end{framed}\vspace{-0.2cm}
	\caption{Description of the main algorithm $A$.} \label{fig:main}
\end{figure}

\begin{proof} Suppose that the true underlying support of $\calD$ is $\mathcal{X}\hspace{-0.03cm} =\hspace{-0.03cm} ( x^{1}, \ldots, x^{s} )$ (as an ordered list). We consider $s$ instances~of~algorithm $\smash{A'_{\text{average-case}}}$
  from~\Cref{sec:reduction}, where each instance has
  parameters $n$, $\tau =  {\delta_{\text{hard}}}/{(2s)}$ and
  $\delta = \delta_{\text{fail}}/({3s})$, 
  and the $i$-th one runs on $T^*$ many traces drawn from $\smash{\calC_{q,q'}(x^{i})}$,
  where
\[
T^*= \exp\left(\Theta\left( \left(\log\hspace{0.04cm} \max\left\{n,\frac{2s}{\delta_{\text{hard}}}\right\}\right)^{1/3}\right)\right)\cdot \log\left(\frac{3s}{\delta_{\text{fail}}}\right)
\]
as specified in \Cref{sec:reduction} (so we have $T=(s/\eps^2)\cdot T^*$).
We say that $\mathcal{X}$ is a \emph{hard support} if either
\begin{flushleft}\begin{enumerate}
\item[(a)] At least one string $x^{i}$, $i\in [s]$, is hard for algorithm $\smash{A'_{\text{average-case}}}$; or
\vspace{-0.1cm}
\item[(b)] After drawing $T$ traces from $\calC_{q,q'}(x^{(i)})$ for each $i\in [s]$,
  $A_{\text{cluster}}$ fails on one of these $\smash{{sT\choose 2}}$ many pairs of traces with probability at least $\delta_{\text{fail}}/3$.
\end{enumerate}\end{flushleft}

We consider a random support $\boldsymbol{\mathcal{X}} = ( \bx^{1},\ldots,\bx^{s} )$ drawn from $\{0,1\}^n$
  independently and uniformly. \Cref{thm:HPPplus} says the probability of a uniform random string being hard
  for $\smash{A'_{\text{average-case}}}$ is~at~most~$ {\delta_{\text{hard}}}/({2s})$. A union bound says the probability our support satisfies (a)
  is at most ${\delta_{\text{hard}}}/{2}$.
On the other~hand, for each support $\mathcal{X}$, we let
  $\lambda(\mathcal{X})$ denote the probability that
  $A_{\text{cluster}}$ fails on at least one of the $\smash{{sT\choose 2}}$ pairs. {\Cref{thm:clustering} implies $\smash{\Ex_{\boldsymbol{\mathcal{X}}}[\lambda(\boldsymbol{\mathcal{X}})] \leq
\smash{{sT\choose 2} \hspace{-0.02cm}\cdot\hspace{-0.02cm} \delta_{\text{cluster}}\le (\delta_{\text{fail}}/3)\hspace{-0.02cm}\cdot\hspace{-0.02cm} \delta_{\text{hard}}/2,}} 
$
where the last inequality follows by setting~the constant hidden
  in the $\Theta(n^{1/3})$ of upper and lower bounds for $s,\eps,\delta_{\text{hard}}$
  and $\delta_{\text{fail}}$ to be sufficiently small (compared to the constant
  hidden in $\delta_{\text{cluster}}$).
By Markov, a random support satisfies (b) with probability at most $\delta_{\text{hard}}/2$.} A union bound on (a) and (b) says the probability of a random support being hard is at most $\delta_{\text{hard}}$.

If $\mathcal{D}'$ is a probability distribution where $\dtv (\mathcal{D}, \mathcal{D}') \leq \eps$, then we say that $\mathcal{D}'$ is $\eps$-accurate. It suffices~to show that for a support $\mathcal{X}$ that is not hard and an arbitrary distribution $\calD$ on that support set, the probability that our algorithm $A$ fails to output an $\eps$-accurate distribution $\calD'$ is at most $\delta_{\text{fail}}$.

Our algorithm has three points of failure. In Step 2, it could fail to cluster the $T$ traces correctly. Given the correct clustering in Step 2, it could fail to learn the underlying string for some cluster in Step 4. Finally, given the correct support, it could fail to output an $\eps$-accurate distribution $\calD'$ in Step 5.

By the definition of hard supports we have that
  Step 2 returns an incorrect clustering with probability at most $\delta_{\text{fail}}/3$.
Given a correct clustering in Step 2, each large $C'_i$ will have at least $T \cdot ({\eps}/{2s})=T^*/\eps\ge T^*$
elements. Since no $\smash{x^{i}}$ is hard for $A'_{\text{average-case}}$, by~\Cref{thm:HPPplus} the probability any instance of $A'_{\text{average-case}}$ fails is at most $\delta_{\text{fail}}/{(3s)}$. By a union bound, the probability of a Step 4 error is at most $ \delta_{\text{fail}}/3$.

Since $\smash{T \geq \Omega(({s}/{\eps^2}) \cdot \log(3 / \delta_{\text{fail}} ))}$ and the large clusters are defined to have size at least a $\eps / 2s$ fraction of the number of traces, then by Corollary~\ref{cor:learn-discrete} with $N = s$, $\kappa = \eps$, $\gamma =  \delta_{\text{fail}}/3$, and $m = T$, given the correct support the probability that Step 5 fails to output an $\eps$-accurate probability distribution is at most $ \delta_{\text{fail}}/3$.
By a union bound, the probability of failure on a support that is not hard is at most $\delta_{\text{fail}}$. 

By~\Cref{thm:clustering} Step 2 takes time $O(nT^2)$. 
By~\Cref{thm:HPPplus} Step 4 takes time $\smash{\text{poly}(n,{s}/{\delta_{\text{hard}}}, \log ({1}/{\delta_{\text{fail}}} ))}$. Step 5 takes time $\smash{O(s)}$ to compute the weights used in $\smash{\mathcal{D}'}$. Therefore, the overall running time of the~algorithm is $\smash{\text{poly}(n, s, 1 / \eps, 1 / \delta_{\text{hard}}, \log({1}/{\delta_{\text{fail}}}))}$. The theorem follows since the sample complexity $T$ is at most $\text{poly}(s, 1 / \eps, \exp(\log^{{1}/{3}} n), \exp( \log^{{1}/{3}} (1 / \delta_{\text{hard}}) ), \log ( 1 / \delta_{\text{fail}} ) ).$
\end{proof}


\begin{flushleft}
\bibliography{allrefs}{}

\newcommand{\etalchar}[1]{$^{#1}$}
\begin{thebibliography}{HMPW08}

\bibitem[ABH14]{ABH14}
Alexandr Andoni, Mark Braverman, and Avinatan Hassidim.
\newblock Phylogenetic reconstruction with insertions and deletions.
\newblock Manuscript, 2014.

\bibitem[ADHR10]{ADHR10}
Alexandr Andoni, Constantinos Daskalakis, Avinatan Hassidim, and
  S{\'{e}}bastien Roch.
\newblock Global alignment of molecular sequences via ancestral state
  reconstruction.
\newblock In {\em {ICS}}, pages 358--369, 2010.

\bibitem[BCF{\etalchar{+}}19]{BCFSS19}
Frank Ban, Xi~Chen, Adam Freilich, Rocco~A. Servedio, and Sandip Sinha.
\newblock Beyond trace reconstruction: Population recovery from the deletion
  channel.
\newblock {\em CoRR}, abs/1904.05532, 2019.

\bibitem[BKKM04]{BKKM04}
T.~Batu, S.~Kannan, S.~Khanna, and A.~McGregor.
\newblock Reconstructing strings from random traces.
\newblock In {\em {Proceedings of the Fifteenth Annual {ACM-SIAM} Symposium on
  Discrete Algorithms, {SODA} 2004}}, pages 910--918, 2004.

\bibitem[Cha19]{Cha19}
Zachary Chase.
\newblock New lower bounds for trace reconstruction.
\newblock {\em arXiv preprint arXiv:1905.03031}, 2019.

\bibitem[DOS17a]{DOS17}
Anindya De, Ryan O'Donnell, and Rocco~A. Servedio.
\newblock Optimal mean-based algorithms for trace reconstruction.
\newblock In {\em Proceedings of the 49th ACM Symposium on Theory of Computing
  (STOC)}, pages 1047--1056, 2017.

\bibitem[DOS17b]{DOS17poprec}
Anindya De, Ryan O'Donnell, and Rocco~A. Servedio.
\newblock Sharp bounds for population recovery.
\newblock {\em CoRR}, abs/1703.01474, 2017.

\bibitem[DR10]{DR10}
Constantinos Daskalakis and S{\'{e}}bastien Roch.
\newblock Alignment-free phylogenetic reconstruction.
\newblock In {\em RECOMB}, pages 123--137, 2010.

\bibitem[DRWY12]{DRWY12}
Z.~Dvir, A.~Rao, A.~Wigderson, and A.~Yehudayoff.
\newblock Restriction access.
\newblock In {\em {Innovations in Theoretical Computer Science}}, pages 19--33,
  2012.

\bibitem[DST16]{DST16}
A.~De, M.~Saks, and S.~Tang.
\newblock Noisy population recovery in polynomial time.
\newblock Technical Report TR-16-026, Electronic Colloquium on Computational
  Complexity, 2016.
\newblock To appear in FOCS 2016.

\bibitem[Fel68]{Feller}
W.~Feller.
\newblock {\em An introduction to probability theory and its applications}.
\newblock John Wiley \& Sons, 1968.

\bibitem[HL18]{HoldenLyons18}
Nina Holden and Russell Lyons.
\newblock Lower bounds for trace reconstruction.
\newblock Available at https://arxiv.org/abs/1808.02336, 2018.

\bibitem[HMPW08]{HMPW08}
T.~Holenstein, M.~Mitzenmacher, R.~Panigrahy, and U.~Wieder.
\newblock Trace reconstruction with constant deletion probability and related
  results.
\newblock In {\em Proceedings of the Nineteenth Annual {ACM-SIAM} Symposium on
  Discrete Algorithms, {SODA} 2008}, pages 389--398, 2008.

\bibitem[HPP18]{HPP18}
Nina Holden, Robin Pemantle, and Yuval Peres.
\newblock Subpolynomial trace reconstruction for random strings and arbitrary
  deletion probability.
\newblock {\em CoRR}, abs/1801.04783, 2018.

\bibitem[Jan18]{Janson18}
Svante Janson.
\newblock Tail bounds for sums of geometric and exponential variables.
\newblock {\em Statistics \& Probability Letters}, 135:1--6, 2018.

\bibitem[KM05]{KM05}
Sampath Kannan and Andrew McGregor.
\newblock More on reconstructing strings from random traces: Insertions and
  deletions.
\newblock In {\em IEEE International Symposium on Information Theory}, pages
  297--301, 2005.

\bibitem[KMMP19]{KMMP19}
Akshay Krishnamurthy, Arya Mazumdar, Andrew McGregor, and Soumyabrata Pal.
\newblock Trace reconstruction: Generalized and parameterized.
\newblock {\em arXiv preprint arXiv:1904.09618}, 2019.

\bibitem[LZ15]{LZ15}
S.~Lovett and J.~Zhang.
\newblock {Improved Noisy Population Recovery, and Reverse Bonami-Beckner
  Inequality for Sparse Functions}.
\newblock In {\em Proceedings of the Forty-Seventh Annual {ACM} on Symposium on
  Theory of Computing, {STOC} 2015, Portland, OR, USA, June 14-17, 2015}, pages
  137--142, 2015.

\bibitem[MPV14]{MPV14}
Andrew McGregor, Eric Price, and Sofya Vorotnikova.
\newblock Trace reconstruction revisited.
\newblock In {\em Proceedings of the 22nd Annual European Symposium on
  Algorithms}, pages 689--700, 2014.

\bibitem[MS13]{MoitraSaks13}
Ankur Moitra and Michael~E. Saks.
\newblock A polynomial time algorithm for lossy population recovery.
\newblock In {\em 54th Annual {IEEE} Symposium on Foundations of Computer
  Science, {FOCS} 2013, 26-29 October, 2013, Berkeley, CA, {USA}}, pages
  110--116, 2013.

\bibitem[NP17]{NazarovPeres17}
Fedor Nazarov and Yuval Peres.
\newblock Trace reconstruction with exp(o(n\({}^{\mbox{1/3}}\))) samples.
\newblock In {\em Proceedings of the 49th Annual {ACM} {SIGACT} Symposium on
  Theory of Computing, {STOC} 2017}, pages 1042--1046, 2017.

\bibitem[OAC{\etalchar{+}}18]{organick2018random}
Lee Organick, Siena~Dumas Ang, Yuan-Jyue Chen, Randolph Lopez, Sergey Yekhanin,
  Konstantin Makarychev, Miklos~Z Racz, Govinda Kamath, Parikshit Gopalan,
  Bichlien Nguyen, et~al.
\newblock Random access in large-scale dna data storage.
\newblock {\em Nature biotechnology}, 36(3):242, 2018.

\bibitem[PSW17]{PSW17}
Yury Polyanskiy, Ananda~Theertha Suresh, and Yihong Wu.
\newblock Sample complexity of population recovery.
\newblock In {\em Proceedings of the 30th Conference on Learning Theory, {COLT}
  2017, Amsterdam, The Netherlands, 7-10 July 2017}, pages 1589--1618, 2017.

\bibitem[PZ17]{PZ17}
Yuval Peres and Alex Zhai.
\newblock Average-case reconstruction for the deletion channel: Subpolynomially
  many traces suffice.
\newblock In {\em FOCS}, pages 228--239, 2017.

\bibitem[VS08]{VS08}
Krishnamurthy Viswanathan and Ram Swaminathan.
\newblock Improved string reconstruction over insertion-deletion channels.
\newblock In {\em Proceedings of the 19th Annual {ACM-SIAM} Symposium on
  Discrete Algorithms}, pages 399--408, 2008.

\bibitem[WY16]{WY16}
A.~Wigderson and A.~Yehudayoff.
\newblock Population recovery and partial identification.
\newblock {\em Machine Learning}, 102(1):29--56, 2016.
\newblock Preliminary version in FOCS~2012.

\bibitem[YGM17]{DNAS}
S.M. Hossein~Tabatabaei Yazdi, Ryan Gabrys, and Olgica Milenkovic.
\newblock Portable and error-free {DNA}-based data storage.
\newblock {\em Scientific Reports}, 7(1):5011, 2017.

\end{thebibliography}
\bibliographystyle{alpha}
\end{flushleft}

\appendix

\section{Deferred proof of Lemma~\ref{lem:keytechnical}} \label{sec:keytechnical}

We recall Lemma~\ref{lem:keytechnical}:

\medskip

\noindent {\bf Lemma~\ref{lem:keytechnical} (restated)} 
\emph{	Let $\rr,\rr'$ be two $n$-patterns drawn independently from $\calR_{n,q,q'}$.
	Then with probability at least $1-\exp(-\Omega(n^{1/3}))$, the
	following three properties all hold:
	\begin{flushleft}\begin{enumerate}
			\item[\emph{(0):}] $|\rr|,|\rr'|\ge \alpha n/2$.\vspace{-0.06cm}
			\item[\emph{(1):}] $|B_\ell(\rr)\cap B_\ell(\rr')|\ge \tau t$ for all $\ell\in [\tilde{s}]$.\vspace{-0.06cm}
			\item[\emph{(2):}] If an $i\in [n]$ appears in both
			$B_\ell(\rr)$ and $B_{\ell'}(\rr')$ for some $\ell,\ell'\in [\tilde{s}]$,
			then we have $\ell=\ell'$.
	\end{enumerate}\end{flushleft}
}
\medskip
We will use the following tail bounds for sums of independent geometric random variables, which are special cases of results proved by~\cite{Janson18}.
\begin{theorem}[Theorems 2.1 and 3.1 in~\cite{Janson18}] \label{thm:janson_geometric_tail}
	Let $p' \in (0,1]$, and $\bX_1,\cdots, \bX_n$ be independent \emph{Geometric}\hspace{0.04cm}$(p')$ random variables. Let $\bX = \sum_{i\in [n]} \bX_i$ and $\mu = \bE[\bX] = n/p'$. Then the following holds:
	\begin{enumerate}
		\item For any $\lambda \geq 1$, we have 
		\[
			\Pr\big[X \geq \lambda \mu\big] \leq \exp\left({-p' \mu (\lambda - 1 - \ln \lambda)}\right).
		\]
		\item For any $0 < \lambda \leq 1$, we have 
		\[
			\Pr\big[X \leq \lambda \mu\big] \leq \exp\left({-p' \mu (\lambda - 1 - \ln \lambda)}\right).
		\]
	\end{enumerate}
\end{theorem}
Note that $\lambda -1 - \ln \lambda \geq 0$ for all $\lambda > 0$, with equality only at $\lambda = 1$. We first derive a simpler expression for the tail bounds, using the following claim:
\begin{claim} \label{clm:Janson_simple}
	Let $f : (-1,\infty) \rightarrow \R$ be defined as
	$
	f(x) = x - \ln (1+x).
	$
	The following  properties hold:
(i) $f(0) = 0$; (ii) $f(x) > x^2/4$ for all $x \in (-1,1]  \setminus \{0\}$; and
(iii) $f(x) \geq x/4$ for all $x \geq 1$.
\end{claim}
\begin{proof}
	The claim follows from elementary calculus. For item (ii) it can be shown that  $g(x) = f(x) - x^2/4$ attains its minimum value $0$ at $x = 0$ and is strictly convex in $(-1,1]$. For item (iii) it is easy to verify that $h(x) = f(x) - x/4$ satisfies $h'(x) > 0$ for all $x \geq 1$, and hence its minimum value is $h(1) \geq 0.05$.
\end{proof}
Letting $x = \lambda - 1$ in~\Cref{thm:janson_geometric_tail}, this claim allows us to replace the $\lambda - 1 - \ln \lambda$ term in the exponent of the tail bounds by either $(\lambda-1)^2 / 4$ or $(\lambda-1)/4$, depending on whether 
$\lambda < 2$ or $\lambda \geq 2$.

Now, we state and prove a few claims that will be useful for proving Lemma   \ref{lem:keytechnical}.
The first claim states that property (0) in Lemma \ref{lem:keytechnical} 
  holds with probability at least $1-\exp(-\Omega(n))$.

\begin{claim} \label{clm_insertion}
With probability at least $1 - \exp(-\Omega(n))$, 
$\br\sim \calR_{n,q,q'}$ satisfies that $|\rr^{(1)}| \geq  \alpha n/2$.
\end{claim}
\begin{proof}
Let $\rr^{(1)}$ be the random string defined earlier in the generation of
	  $\smash{\br\sim\calR_{n,q,q'}}$. 
	As $|\rr^{(1)}|$ is a sum of $n$ independent Geometric\hspace{0.04cm}$(p')$ random variables, we have $\smash{\mu = \bE [\hspace{0.04cm} |\rr^{(1)}|\hspace{0.04cm}] = {n}/{p'}}$. Invoking~\Cref{thm:janson_geometric_tail} with $\lambda = 3/4$ and Part (1) of Claim \ref{clm:Janson_simple} with $x = \lambda-1$,  the probability of 
	$|\rr^{(1)}|<3n/(4p')$ is $\exp(-\Omega(n))$.
	
Fixing any realization $r^{(1)}$ of $\br^{(1)}$ with $|r^{(1)}|\ge 3n/(4p')$,
  it follows from the standard Chernoff bound that the probability of 
  $\smash{|\br|<\alpha n/2\le (2p)/3\cdot |r^{(1)}|}$ is at most $\exp(-\Omega(n))$. 
This finishes the proof.
\end{proof}



Fix an $i \in [n]$. Let $\bY^*_i$ be the random variable denoting the number of characters before $i$ in $\rr^{(1)}$ (after insertions). Recall that $\bY_i$ denotes the number of characters before $i$ in $\rr^{(1)}$ that survive deletions; note that $\bY_i$ is well-defined even if $i$ is deleted. Then $\bE[\bY^*_i] = (i/p') - 1$, and $\bE[\bY_i] = p \cdot\bE[\bY^*_i] = ((i/p')-1)p$.

\begin{lemma} \label{lem_coordinate_deviation_t}
For any $i\in [n]$, the probability that $\big|\bY_i-\bE[\bY_i]\big|\ge 0.05\hspace{0.02cm}t$
  is at most 
$\exp (-\Omega (n^{1/3} ))$.
\end{lemma}

\begin{proof}
Let $\eps=0.05$ in the proof. We have 
$$
\big|\bY_i-\bE[\bY_i]\big|
\le \big|\bY_i- p \bY_i^*\big| + \big| p \bY_i^*-p \bE[\bY_i^*]\big|
=\big|\bY_i-p \bY_i^*\big| +p \cdot \big|\bY_i^*-\bE[\bY_i^*]\big|.
$$
We first show that $|\bY_i^*-\bE[\bY_i^*]|\le  \eps t/(2p)$ with probability 
  at least $1-\exp(-\Omega(n^{1/3}))$.
Next  conditioning on any fixed realization $r^{(1)}$ of $\rr^{(1)}$ with
  $|Y_i^*-\bE[\bY_i^*]|\le \eps t/(2p)$ (in particular this implies that 
  $Y_i^*=O(n)$) we show that $|\bY_i-pY_i^*|\le \eps t/2$ with probability at least
  $1-\exp(-\Omega(n^{1/3}))$.
The lemma then follows by combining these two steps.
Given that the second step follows from the Hoeffding bound 
  (with $Y_i^*=O(n)$ and $t=n^{2/3}$), we focus on the first part in the rest of the proof. 


	First we analyze the lower tail, i.e., the probability of 
	$\bY_i^*-\bE[\bY_i^*]\le -\eps t/(2p)$.
Because $\bY^*_i \geq 0$~we may assume $\bE[\bY^*_i] > \eps t / (2p)$ (otherwise $\bY^*_i \geq \bE[\bY^*_i] - \eps t / (2p)$ trivially).
	Let $$\lambda = 1 - \frac{\eps t}{2p \bE[\bY^*_i]}\quad \text{and}\quad x = \lambda - 1 = - \frac{\eps t}{2p \bE[\bY^*_i]},$$ so that $\lambda \bE[\bY^*_i] = \bE[\bY^*_i] - {\eps t}/({2p})$. By~\Cref{thm:janson_geometric_tail} and Part (1) of Claim \ref{clm:Janson_simple}, we have
	\[
		\Pr\left[\bY^*_i\le \bE[\bY^*_i] - \frac{\eps t}{2p}\right] \leq \exp\left(-\Omega\left(\bE[\bY^*_i]\cdot  \frac{t^2}{ \bE[\bY^*_i]^2}\right)\right) \leq \exp\left(-\Omega\left(\frac{t^2}{n}\right)\right) = \exp\left(-\Omega(n^{1/3})\right).
	\]
	For the second inequality, we used the fact that $\bE[\bY^*_i] = O(n)$. Similarly, we analyze the upper tail. Let $$\lambda = 1 + \frac{\eps t}{2p \bE[\bY^*_i]}\quad \text{and}\quad x = \lambda - 1 = \frac{\eps t}{2p \bE[\bY^*_i]}.$$ If $\lambda \leq 2$,~\Cref{thm:janson_geometric_tail} and Part (1) of Claim \ref{clm:Janson_simple} imply that
	\[
	\Pr\left[\bY^*_i \ge \bE[\bY^*_i] + \frac{\eps t}{2p}\right] \leq \exp\left(-\Omega\left(\bE[\bY^*_i]\cdot  \frac{t^2}{ \bE[\bY^*_i]^2}\right)\right) \leq \exp\left(-\Omega\left(\frac{t^2}{n}\right)\right) = \exp\left(-\Omega(n^{1/3})\right).
	\]
	On the other hand, if $\lambda \geq 2$, then $x \geq 1$. By~\Cref{thm:janson_geometric_tail} and Part (2) of Claim \ref{clm:Janson_simple}, we have
	\[
	\Pr\left[\bY^*_i\ge \bE[\bY^*_i] + \frac{\eps t}{2p}\right] \leq \exp\left(-\Omega\left(\bE[\bY^*_i]\cdot  \frac{t}{\bE[\bY^*_i]}\right)\right) \leq \exp\left(-\Omega(t)\right) = \exp\left(-\Omega(n^{2/3})\right).
	\]
This finishes the proof of the lemma.
	\vspace{-0.1cm}
\end{proof}

We are ready to prove Lemma \ref{lem:keytechnical}.\medskip

\begin{proof}[Proof of Lemma \ref{lem:keytechnical}] 
We work on the three events separately and apply a union bound at the end.
	\begin{flushleft}\begin{enumerate}
			\item[\emph{(0):}] It follows from Claim \ref{clm_insertion}
			that property (0) holds with probability at least $1-\exp(-\Omega(n))$.
			
			\item[\emph{(1):}] 	Fix $\ell \in [\tilde{s}]$. Recall that $I_\ell=[(2\ell-2)t+1, (2\ell-1)t] \cap \Z$. Let 
			\[
				I_\ell^{(0)} := \left[\frac{p'}{p}(2\ell-1.9)t+1, \frac{p'}{p}(2\ell-1.1)t \right] \cap \Z.
			\]
			Then $I_\ell^{(0)} \subset [n]$. We will show that with probability at least $1 - \exp(-\Omega(n^{1/3}))$, 
			 both properties below hold:
			\begin{enumerate}
				\item At least $0.7t p p'$ elements in $I_\ell^{(0)}$ survive in both 
				  $\rr$ and $\rr'$;\vspace{-0.05cm}
				\item If an element $i \in I_\ell^{(0)}$ survives in both $\rr$ and $\rr'$, then $i \in B_\ell(\rr) \cap B_\ell(\rr')$.
			\end{enumerate}
Given that (a) and (b) together imply property (1), we have that 
  property (1) holds for $\ell\in [\tilde{s}]$ with probability at least $1-\exp(-\Omega(n^{1/3}))$.
A union bound over all $\ell \in [\tilde{s}]$ implies that property (1) holds for all $\ell\in [\tilde{s}]$ with probability at least $1 - \tilde{s} \cdot  \exp(-\Omega(n^{1/3})) = 1 - \exp(-\Omega(n^{1/3})).$ 
			

So it suffices to show that (a) and (b) happen with probability
  at least $1-\exp(-\Omega(n^{1/3}))$.
  For (a), it follows from a standard Chernoff bound that (a) holds 
  with probability at least $\smash{1-\exp(-\Omega(n^{2/3}))}$.			
			For (b), 
			let $i_0$ and $i_1$ be the left and right endpoints of $\smash{I_\ell^{(0)}}$, respectively. 
			Let $\bY_{i_0}, \bY_{i_1}$ ($\bY_{i,0}',\bY_{i_1}'$) be as defined earlier with respect to $\rr$ ($\rr'$). Note that $\bE[\bY_{i_0}] = ((i_0 / p') - 1) p$ 
			and $\bE[\bY_{i_1}] = ((i_1/p')-1)p$.  
			Then by Lemma~\ref{lem_coordinate_deviation_t} (and a union bound),  
			with probability at least $1 - 4\exp(-\Omega(n^{1/3}))$, we have:
			\[\bY_{i_0} \geq \bE[\bY_{i_0}] - 0.05 t > (2 \ell - 2)t\quad\text{and}\quad
\bY_{i_1} \leq \bE[\bY_{i_1}] + 0.05 t < (2 \ell - 1)t,\]
and the same holds for $\bY_{i_0}'$ and $\bY_{i_1}'$.
			When all these events occur, then clearly all characters in $I_\ell^{(0)}$ that survive in $\rr, \rr'$ are in $I_\ell$ in both $n$-patterns. 
	This finishes the analysis of property (1).  			
			
			\item[\emph{(2):}] Suppose a character $i \in [n]$ appears in $B_\ell(\rr)$ and $B_{\ell'}(\rr')$ for some $\ell \neq \ell'$. Let $\bY_i, \bY'_i$ denote the number of characters before $i$ in $\rr, \rr'$ respectively. Then $\bE[\bY_i] = \bE[\bY'_i]$. As any two distinct blocks are separated by at least $t$ positions in the $n$-patterns, we have $|\bY_i - \bY'_i| \geq t$. Triangle inequality implies that $\smash{|\bY_i - \bE[\bY_i]| \geq t/2}$ or $\smash{|\bY'_i - \bE[\bY'_i]| \geq t/2}$. Assume without loss of generality that $|\bY_i - \bE[\bY_i]| \geq t/2>0.05\hspace{0.02cm}t$. Instantiating Lemma \ref{lem_coordinate_deviation_t}, we conclude that this event happens with probability at most $n\cdot \exp(-\Omega(n^{1/3}))$ which remains $\exp(-\Omega(n^{1/3}))$.
	\end{enumerate}\end{flushleft}
The lemma follows from a union bound.
\end{proof}

\end{document}